\documentclass{amsart}
\usepackage{amsfonts}
\usepackage{graphicx}

\newtheorem{theorem}{Theorem}
\theoremstyle{plain}

\numberwithin{equation}{section}
\input{tcilatex}

\begin{document}
\title[Decomposition of Third-order Linear Time-varying Systems]{%
Decomposition of Third-order Linear Time-varying Systems into its Second and
First-order Commutative Pairs}
\author{Mehmet Emir Koksal}
\address[Mehmet Emir Koksal]{Department of Mathematics, Ondokuz Mayis
University \\
52139 Atakum, Samsun, Turkey }
\email[Mehmet Emir Koksal]{emir\_koksal@hotmail.com}
\author{Ali Yakar}
\address[Ali Yakar]{Department of Mathematics, Gaziosmanpasa University \\
60250 Tokat, Turkey }
\date{September 15, 2017}
\subjclass[2010]{34H05, 49K15, 93B52, 93C15}
\keywords{Differential equations, initial conditions, analogue control,
equivalent circuits, physical systems}

\begin{abstract}
Decomposition is a common tool for synthesis of many physical systems. It is
also used for analyzing large scale systems which then known as fearing and
reconstruction. On the other hand, commutativity of cascade connected
systems have gained a grate deal of interest and its possible benefits have
been pointed out in the literature. In this paper, the necessary and
sufficient conditions for decomposability of a third-order linear
time-varying systems as a pair of second and first-order systems of which
parameters are also explicitly expressed. Further, additional requirements
in case of non-zero initial conditions are derived.This paper highlights the
direct formulas for realization of any third order linear time-varying
system as a series (cascade) connection of first and second order
subsystems. This series connection is commutative so that it is independent
from the sequence of subsystems in the connection. Hence, the convenient
sequence can be decided by considering the overall performance of the system
when the sensitivity, disturbance and robustness effects are considered.
Realization covers transient responses as well as steady state responses.
\end{abstract}

\maketitle

\section{Introduction}

\bigskip Differential equations arise as common models in the physical,
mathematical, biological and engineering sciences and most real physical
processes are governed by differential equations. The fundamental laws
governing many physical process are known relationships between various
quantities and their derivatives. In general, most real physical processes
involve more than one independent variable and the corresponding
differential equations. Especially, differential equations are used for
modelling problems in electric-electronics engineering, the touchstone and
largest branch of engineering technology and includes a diverse range of
sub-disciplines, such as embedded systems, control systems,
telecommunications, and power systems. For instance, in system and control
theory, the transfer function, also known as the system function or network
function, is a mathematical representation of the relation between the input
and output based on the differential equations describing the system such as
cascade and feedback connections. When the cascade connection in system
design is considered, the commutativity concept places an prominent role to
improve different system performances.

Cascade connection of subsystems is a commonly used method for designing
many engineering systems especially electrical and electronic devices \cite%
{1,2,3,4,5}. For example, cascade connection is used for connecting the
server module located in another subnetwork via an intermediate computer
that has two network interfaces for two subnetworks. The order of connection
is important for achieving more reliable systems which are less sensitive
and more robust to internal and external disturbances, and it may depend on
many criteria such as the used design technique, engineering ingenuity, and
traditional habits. Therefore, the change of the order of connection may be
thought for the possibility of obtaining better performances without
spoiling the main function of the total system (commutativity). Hence,
commutativity is important from engineering point of view.

When two simple systems are connected in cascade, that is the output of the
former acts as the input of the later \cite{6,7}, if the order of connection
does not change the input-output relation of the combined system then we say
that these systems are commutative.

There are a great deal of literature about the commutativity of linear
continuous time-varying systems \cite[9-16]{8} though there are a few works
on the discrete time-varying systems \cite{17,18}. The first paper on the
commutativity in the literature has been studied by Marshall in \cite{8} and
it is proved that a time-varying system can be commutative with another
time-varying system. Then, commutativity conditions of second-order,
third-order and fourth-order systems were studied in \cite{9,10,11}, \cite%
{12} and \cite{13}, respectively. In \cite{14}, the most general necessary
and sufficient conditions for the commutativity of systems of any order but
without initial conditions were studied. This study also includes results
concerning the commutativity properties of feed-back control systems and
Euler differential systems. Moreover, the previous results for commutativity
conditions of first-order, second-order, third-order and fourth-order
systems were shown to be deduced from the main theorem of \cite{14}.

More than two decades later, the explicit commutativity conditions for
linear time-varying differential systems with non-zero initial conditions
\cite{15} and the explicit commutativity conditions for the fifth-order
systems derived for the first time in \cite{15}.

Final study on the commutativity of analogue systems was studied in \cite{16}
covering necessary and sufficiently conditions for the decomposition of a
second-order linear time-varying system into two cascade connected
commutative first-order linear time-varying subsystems. Further, explicit
formulas describing these subsystems were presented by illustrative examples
and simulations.

References \cite{17,18} are debuted to investigation of commutativity of
discrete-time (digital) systems. The concept of commutativity for digital
systems was defined in \cite{17} for the first time. Then, the possible
benefits of commutativity such as noise disturbance,effects, parameter
sensitivity are outlined in \cite{18}. In unpublished work, the transitivity
property is examined and it holds for analog systems. For digital systems
however, it has not been reported anywhere.

In this paper, after deriving some mathematical preliminaries in Section II,
the basic equations are that must be satisfied for commutativity are derived
for in Section III. These equations are solved in Section IV. In Section V,
the coefficients of the second and first-order components are explicitly
expressed in terms of those of the original third-order system. Section VI
covers a few illustrative examples. And finally, the paper ends with Section
VII conclusion.

\section{Mathematical Preliminaries}

Let $C$ be a third-order linear time-varying analog system described by
\begin{equation}
c_{3}(t)y^{\prime \prime \prime }(t)+c_{2}(t)y^{\prime \prime
}(t)+c_{1}(t)y^{\prime }(t)+c_{0}(t)y(t)=x(t),  \tag{2.1}  \label{1}
\end{equation}%
with the input $x(t)$ and output $y(t)$. Where $c_{i}(t)$ are time-varying
coefficients which are piecewise continuous on $[t_{0},\infty )$; this set
of function are devoted by $P[t_{0},\infty )$; also assume the initial
conditions $y(t_{0})$, $y^{\prime }(t_{0})$, $y^{\prime \prime }(t_{0})$ at
the initial time $t_{0}$ $\in R,$ where the number of overhaead dots
represent the order of derivatives. Due to its order of $3$, $c_{3}(t)\neq 0$%
.

It is well-known that such a system has a unique solution for all $x(t)\in
P[t_{0},\infty )$. Consider the decomposition of $C$ as the cascade
connection of a first-order system $A$ and second-order $B$ described by%
\begin{equation}
A:a_{1}(t)y_{A}^{\prime }(t)+a_{0}(t)y_{A}(t)=x_{A}(t),  \tag{2.2}
\label{2a}
\end{equation}%
\begin{equation}
B:b_{2}(t)y_{B}^{\prime \prime }(t)+b_{1}(t)y_{B}^{\prime
}(t)+b_{0}(t)y_{B}(t)=x_{B}(t),  \tag{2.3}  \label{2b}
\end{equation}%
with the initial conditions
\begin{equation}
y_{A}(t_{0})  \tag{2.4}  \label{2c}
\end{equation}%
\begin{equation}
y_{B}(t_{0}),\text{ }y_{B}^{\prime }(t_{0}).  \tag{2.5}  \label{2d}
\end{equation}%
Due to their orders $a_{1}(t)\neq 0$, $b_{2}(t)\neq 0$. Further, assume $%
a_{i}$, $b_{i}$, $x_{A}$, $x_{B}\in P[t_{0},\infty )$. Moreover, assume that
$a_{i}$'s are differentiable up to second-order and $b_{i}^{\prime }s$ are
differentiable up to first-order. Assume also that the cascade connection of
$A$ and $B$, denoted by $AB$ or $BA$ according to their order of connection
as shown in Fig. $1a$ and $1b$ respectively, are commutative. That is $AB$
and $BA$ has the same input-output relation.

\begin{figure}[tbp]
\includegraphics[width=10cm, height=6cm]{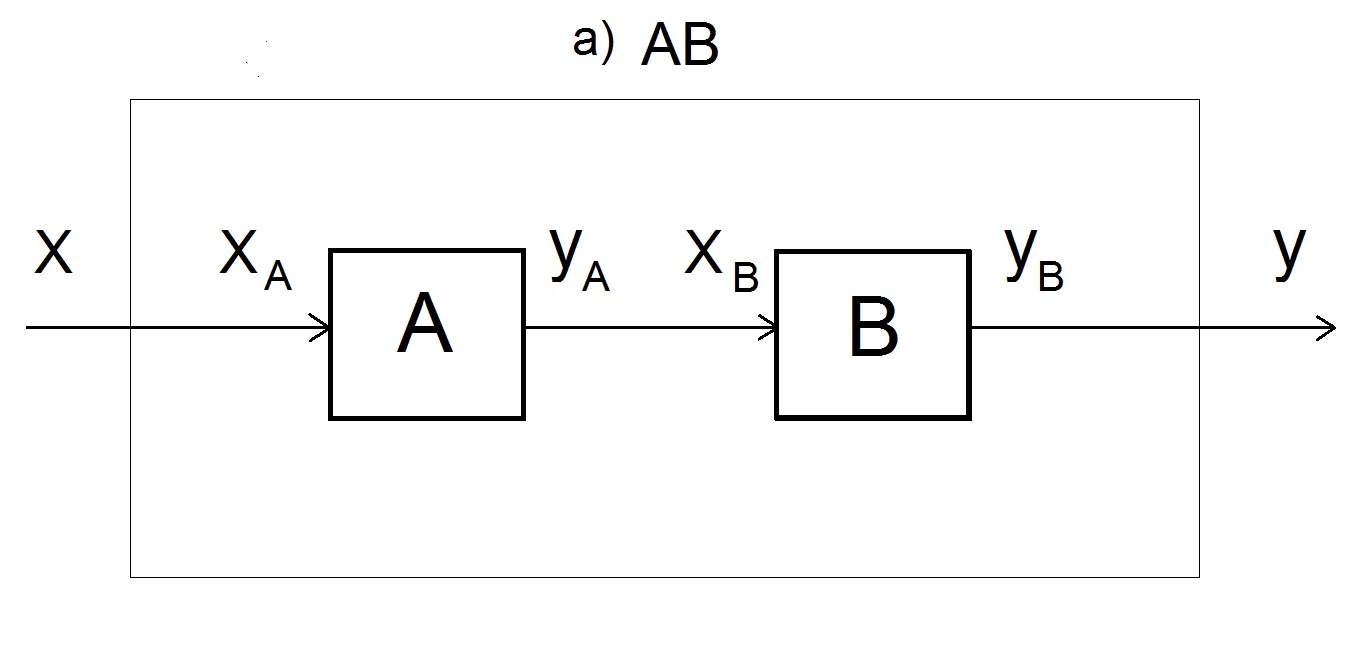} \includegraphics[width=10cm,
height=6cm]{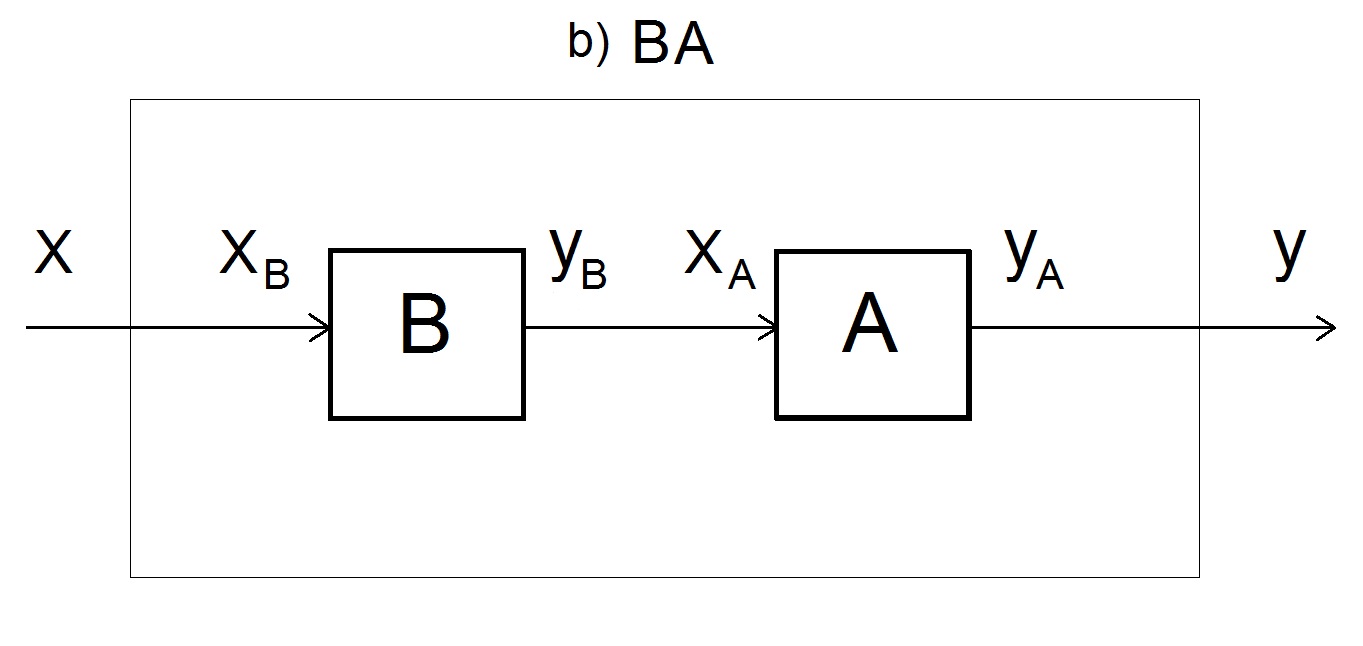}
\caption{Cascade connection of $A$ and $B;$ $a$) $AB$, $b$) $BA$}
\label{.}
\end{figure}

\bigskip

Due to the connection in Fig. $1a$, it is obvious that
\begin{equation}
x_{A}(t)=x(t),  \tag{2.6}  \label{3a}
\end{equation}%
\begin{equation}
y_{A}(t)=x_{B}(t),  \tag{2.7}  \label{3b}
\end{equation}%
\begin{equation}
y_{B}(t)=y(t).  \tag{2.8}  \label{3c}
\end{equation}

Differentiating Eq. (\ref{2b}), we obtain

\begin{equation}
b_{2}^{\prime }y_{B}^{\prime \prime }+b_{2}y_{B}^{\prime \prime \prime
}+b_{1}^{\prime }y_{B}^{\prime }+b_{1}^{\prime }y_{B}^{\prime \prime
}+b_{0}^{\prime }y_{B}+b_{0}y_{B}^{\prime }=x_{B}^{\prime }.  \tag{2.9}
\label{4a}
\end{equation}%
From Eq. (\ref{3b}), $x_{B}^{\prime }=y_{A}^{\prime },$ and then solving Eq.
(\ref{2a}) for $y_{A}^{\prime },$ finally using Eq. (\ref{3b}), we again
obtain%
\begin{equation*}
x_{B}^{\prime }=y_{A}^{\prime }=\frac{x_{A}-a_{0}y_{A}}{a_{1}}=\frac{%
x_{A}-a_{0}x_{B}}{a_{1}}
\end{equation*}%
\begin{equation}
=\frac{1}{a_{1}}\left[ x_{A}-a_{0}(b_{2}y_{B}^{\prime \prime
}+b_{1}y_{B}^{\prime }+b_{0}y_{B})\right] ,  \tag{2.10}  \label{4b}
\end{equation}%
where the last equality is obtained by using the expression Eq. (\ref{2b})
for $x_{B}$. Finally, inserting Eq. (\ref{4b}) in Eq. (\ref{4a}) and
replacing $y_{B}\rightarrow y$, $x_{A}\rightarrow x$ due to Eqs. (\ref{3c})
and (\ref{3a}), respectively, we obtain the following third-order
differential system for the connection $AB$%
\begin{equation*}
a_{1}b_{2}y^{\prime \prime \prime }+(a_{1}b_{2}^{\prime
}+a_{1}b_{1}+a_{0}b_{2})y^{\prime \prime }
\end{equation*}%
\begin{equation}
+(a_{1}b_{1}^{\prime }+a_{1}b_{0}+a_{0}b_{1})y^{\prime }+a_{1}b_{0}^{\prime
}+a_{0}b_{0})y=x,  \tag{2.11}  \label{5a}
\end{equation}%
\begin{equation}
y(t_{0})=y_{B}(t_{0}),  \tag{2.12}  \label{5b}
\end{equation}%
\begin{equation}
y^{\prime }(t_{0})=y_{B}^{\prime }(t_{0}),  \tag{2.13}  \label{5c}
\end{equation}%
\begin{equation}
y^{\prime \prime }(t_{0})=y_{B}^{\prime \prime }(t_{0})=\frac{%
y_{A}(t_{0})-b_{0}(t_{0})y_{B}(t_{0})-b_{1}(t_{0})y_{B}^{\prime }(t_{0})}{%
b_{2}(t_{0})}.  \tag{2.14}  \label{5d}
\end{equation}%
Eqs. (\ref{5b}) and (\ref{5c}) are obvious due to Eq. (\ref{3c}). Eq. (\ref%
{5d}) is obtained as follows: Due to Eq. (\ref{3c}), $y(t_{0})=y_{B}^{\prime
\prime }(t_{0})$ which is computed from Eq. (\ref{2b}) and inserting $%
x_{B}(t_{0})=y_{A}(t_{0})$ due to Eq. (\ref{3b}).

Similarly, due to the connection in Fig. $1b$, it is obvious that
\begin{equation}
x_{B}(t)=x(t),  \tag{2.15}  \label{6a}
\end{equation}%
\begin{equation}
y_{B}(t)=x_{A}(t),  \tag{2.16}  \label{6b}
\end{equation}%
\begin{equation}
y_{A}(t)=y(t).  \tag{2.17}  \label{6c}
\end{equation}

Differentiating (\ref{2a}) two times and ordering the terms, we obtain%
\begin{equation}
a_{1}y_{A}^{\prime \prime \prime }+(2a_{1}^{\prime }+a_{0})y_{A}^{\prime
\prime }+(a_{1}^{\prime \prime }+2a_{0}^{\prime })y_{A}^{\prime
}+a_{0}^{\prime \prime }y_{A}=x_{A}^{\prime \prime }.  \tag{2.18}  \label{7a}
\end{equation}%
Since $x_{A}^{\prime \prime }(t)=y_{B}^{\prime \prime }(t)$ due to Eq. (\ref%
{6b}), finding $y_{B}^{\prime \prime }$ from Eq. (\ref{2b}), and using Eq. (%
\ref{6b}) again, we have%
\begin{equation}
x_{A}^{\prime \prime }=y_{B}^{\prime \prime }=\frac{x_{B}-b_{1}y_{B}^{\prime
}-b_{0}y_{B}}{b_{2}}=\frac{x_{B}-b_{1}x_{A}^{\prime }-b_{0}x_{A}}{b_{2}}.
\tag{2.19}  \label{7b}
\end{equation}

Next inserting in the value of $x_{A}$ from Eq. (\ref{2a}) and the value of $%
x_{A}^{\prime }$ from derivative of Eq. (\ref{2a}) into the above equation,
we obtain
\begin{equation}
x_{A}^{\prime \prime }=\frac{x_{B}-b_{1}(a_{1}^{\prime }y_{A}^{\prime
}+a_{1}y_{A}^{\prime \prime }+a_{0}^{\prime }y_{A}^{\prime
}+a_{0}y_{A}^{\prime })-b_{0}(a_{1}y_{A}^{\prime }+a_{0}y_{A})}{b_{2}}.
\tag{2.20}  \label{7c}
\end{equation}

Inserting Eq. (\ref{7c}) in (\ref{7a}) and noting $y_{A}=y$ (Eq. (\ref{6c}))
and $x_{B}=x$, (Eq. (\ref{6a})), we obtain the third-order differential
equation describing $BA$ as%
\begin{equation*}
a_{1}b_{2}y^{\prime \prime \prime }+(2a_{1}^{\prime
}b_{2}+a_{0}b_{2}+a_{1}b_{1})y^{\prime \prime }
\end{equation*}%
\begin{equation}
+(a_{1}^{\prime \prime }b_{2}+2a_{0}^{\prime }b_{2}+a_{1}^{\prime
}b_{1}+a_{0}b_{1}+a_{1}b_{0})y^{\prime }+(a_{0}^{\prime \prime
}b_{2}+a_{0}^{\prime }b_{1}+a_{0}b_{0})y=x,  \tag{2.21}  \label{8a}
\end{equation}%
\begin{equation}
y(t_{0})=y_{A}(t_{0}),  \tag{2.22}  \label{8b}
\end{equation}%
\begin{equation}
y^{\prime }(t_{0})=y_{A}^{\prime }(t_{0})=\frac{%
y_{B}(t_{0})-a_{0}(t_{0})y_{A}(t_{0})}{a_{1}(t_{0})},  \tag{2.23}  \label{8c}
\end{equation}%
\begin{equation*}
y^{\prime \prime }(t_{0})=y_{A}^{\prime \prime }(t_{0})=\frac{1}{a_{1}(t_{0})%
}y_{B}^{\prime }(t_{0})-\frac{a_{0}(t_{0})+a_{1}^{\prime }(t_{0})}{%
a_{1}^{2}(t_{0})}y_{B}(t_{0})
\end{equation*}%
\begin{equation}
+\left[ \frac{a_{0}^{2}(t_{0})+a_{1}^{\prime }(t_{0})a_{0}(t_{0})}{%
a_{1}^{2}(t_{0})}-\frac{a_{0}^{\prime }(t_{0})}{a_{1}(t_{0})}\right]
y_{A}(t_{0}).  \tag{2.24}  \label{8d}
\end{equation}

The derivative of the initial conditions in Eqs. (\ref{8b})-(\ref{8d}) is
done as follows: Eq. (\ref{8b}) follows from Eq. (\ref{6c}). To find Eq. (%
\ref{8b}), we start from Eq. (\ref{6c}) and write $y^{\prime
}(t)=y_{A}^{\prime }(t_{0})$, from Eq. (\ref{2a})%
\begin{equation}
y(t)=y_{A}^{\prime }(t)=\frac{x_{A}(t)-a_{0}(t)y_{A}(t)}{a_{1}(t)}=\frac{%
y_{B}(t)-a_{0}(t)y_{A}(t)}{a_{1}(t)}.  \tag{2.25}  \label{9a}
\end{equation}%
Inserting $t=t_{0}$ yields Eq. (\ref{8c}). To find Eq. (\ref{8d}), we start
from Eq. (\ref{6c}), take derivative of Eq. (\ref{2a}) and solve result for $%
y_{A}^{\prime \prime }$%
\begin{equation}
y^{\prime \prime }=y_{A}^{\prime \prime }=\frac{y_{B}^{\prime
}-(a_{1}^{\prime }+a_{0})y_{A}^{\prime }-a_{0}^{\prime }y_{A}}{a_{1}}.
\tag{2.26}  \label{9b}
\end{equation}

Using the expression Eq. (\ref{9a}) for $y_{A}^{\prime }$ in Eq. (\ref{9b}),
ordering the terms and evaluating at $t=t_{0}$ yields the initial conditions
in Eq. (\ref{8d}).

\section{Commutativity Requirements}

For the commutativity of subsystem $A$ and $B$, their combinations $AB$ and $%
BA$ must have the same outputs for general values of the same input and the
same initial conditions. This is due to the existence of unique equal
solutions of differential equations derived in Eqs. (\ref{5a})-(\ref{5d})
and (\ref{8a})-(\ref{8d}) for the same input and initial conditions. Hence,
equating the coefficients of these differential equations, collecting the
like terms we result with%
\begin{equation}
a_{1}b_{2}^{\prime }=2a_{1}^{\prime }b_{2}  \tag{3.1}  \label{10a}
\end{equation}%
\begin{equation}
a_{1}b_{1}^{\prime }=a_{1}^{\prime }b_{1}+(a_{1}^{\prime \prime
}+2a_{0}^{\prime })b_{2}  \tag{3.2}  \label{10b}
\end{equation}%
\begin{equation}
a_{1}b_{0}^{\prime }=a_{0}^{\prime \prime }b_{2}+a_{0}^{\prime }b_{1}
\tag{3.3}  \label{10c}
\end{equation}%
\begin{equation}
y=y_{B}=y_{A}  \tag{3.4}  \label{11a}
\end{equation}%
\begin{equation}
y^{\prime }=y_{B}^{\prime }=\frac{y_{B}-a_{0}y_{A}}{a_{1}},  \tag{3.5}
\label{11b}
\end{equation}%
\begin{equation}
y^{\prime \prime }=\frac{y_{A}-b_{0}y_{B}-b_{1}y_{B}^{\prime }}{b_{2}}=\frac{%
1}{a_{1}}y_{B}^{\prime }-\frac{a_{0}+a_{1}^{\prime }}{a_{1}^{2}}y_{B}+\left(
\frac{a_{0}^{2}+a_{0}a_{1}^{\prime }-a_{0}^{\prime }a_{1}}{a_{1}^{2}}\right)
y_{A}  \tag{3.6}  \label{11c}
\end{equation}

Note that Eqs. (\ref{11a})-(\ref{11c}) (so should (\ref{12a})-(\ref{12d}))
should be valid at the initial time $t=t_{0}$ which is not shown explicitly.
Before proceeding further we simplify Eqs. (\ref{11a})-(\ref{11c}) to obtain
simpler set of constraints.

\begin{equation}
y=y_{B}=y_{A},  \tag{3.7}  \label{12a}
\end{equation}

\begin{equation}
y^{\prime }=y_{B}^{\prime }=\frac{1-a_{0}}{a_{1}}y_{A},  \tag{3.8}
\label{12b}
\end{equation}

\begin{equation}
y^{\prime \prime }=\left[ \frac{1-b_{0}}{b_{2}}-\frac{b_{1}(1-a_{0})}{%
b_{2}a_{1}}\right] y_{A}=\frac{(1-a_{0})(1-a_{0}-a_{1}^{\prime
})-a_{0}^{\prime }a_{1}}{a_{1}^{2}}y_{A}.  \tag{3.9}  \label{12c}
\end{equation}%
Hence, Eq. (\ref{12c}) requires.%
\begin{equation}
\left[ \frac{1-b_{0}}{b_{2}}-\frac{b_{1}(1-a_{0})}{b_{2}a_{1}}-\frac{%
(1-a_{0})(1-a_{0}-a_{1}^{\prime })-a_{0}^{\prime }a_{1}}{a_{1}^{2}}+\frac{%
a_{0}^{\prime }}{a_{1}}\right] y_{A}=0.  \tag{3.10}  \label{12d}
\end{equation}

\section{Explicit Commutativity Requirements}

Eq. (\ref{10a}) has a solution for $b_{2}$ in terms of $a_{i}$'s as%
\begin{equation}
b_{2}=e_{2}a_{1}^{2},  \tag{4.1}  \label{13a}
\end{equation}%
where $c_{2}$ is an arbitrary non-zero constant. Using this solution in (\ref%
{10b}) and taking integral, we proceed

\begin{equation*}
a_{1}b_{1}^{\prime }=a_{1}^{\prime }b_{1}+(a_{1}^{\prime \prime
}+2a_{0}^{\prime })e_{2}a_{1}^{2}
\end{equation*}

\begin{equation*}
\frac{a_{1}b_{0}^{\prime }-a_{1}^{\prime }b_{1}}{a_{1}^{2}}%
=e_{2}(a_{1}^{\prime \prime }+2a_{0}^{\prime })
\end{equation*}

\begin{equation*}
\frac{d}{dt}\left( \frac{b_{1}}{a_{1}}\right) =e_{2}(a_{1}^{\prime \prime
}+2a_{0}^{\prime })
\end{equation*}

\begin{equation*}
\frac{b_{1}}{a_{1}}=e_{2}(a_{1}^{\prime }+2a_{0})+e_{1}
\end{equation*}

\begin{equation}
b_{1}=e_{2}(a_{1}^{\prime }+2a_{0})a_{1}+e_{1}a_{1}.  \tag{4.2}  \label{13b}
\end{equation}%
Inserting values of $b_{2}$ in Eq. (\ref{13a}) and $b_{1}$ in Eq. (\ref{13b}%
) into Eq. (\ref{10c}), we proceed

\begin{equation*}
a_{1}b_{0}^{\prime }=a_{1}^{\prime \prime }e_{2}a_{1}^{2}+a_{0}^{\prime }
\left[ e_{2}(a_{1}^{\prime }+2a_{0})a_{1}+e_{1}a_{1}\right]
\end{equation*}%
\begin{equation*}
b_{0}^{\prime }=a_{1}^{\prime \prime }e_{2}a_{1}+a_{0}^{\prime }\left[
e_{2}(a_{1}^{\prime }+2a_{0})+e_{1}\right]
\end{equation*}

\begin{equation*}
=e_{2}\left( a_{1}^{\prime \prime }+a_{0}^{\prime }a_{1}^{\prime
}+2a_{0}a_{0}^{\prime }\right) +e_{1}a_{0}^{\prime }=e_{2}\frac{d}{dt}%
(a_{0}^{\prime }a_{1}+a_{0}^{2})+e_{1}a_{0}^{\prime },
\end{equation*}

\begin{equation}
b_{0}=e_{2}(a_{0}^{\prime }a_{1}+a_{0}^{2})+e_{1}a_{0}^{\prime }+e_{0},
\tag{4.3}  \label{13c}
\end{equation}%
where $e_{0}$ is an integration constant. In the matrix form
\begin{equation}
\left[
\begin{array}{c}
b_{2} \\
b_{1} \\
b_{0}%
\end{array}%
\right] =\left[
\begin{array}{ccc}
a_{1}^{2} & 0 & 0 \\
a_{0}^{\prime }+2a_{0} & 1 & 0 \\
a_{0}^{\prime }a_{1}+a_{0}^{2} & a_{0}^{\prime } & 1%
\end{array}%
\right] \left[
\begin{array}{c}
e_{2} \\
e_{1} \\
e_{0}%
\end{array}%
\right] .  \tag{4.4}  \label{14}
\end{equation}

Hence, Eqs. (\ref{10a})-(\ref{10c}) are equivalently replaced by Eq. (\ref%
{14}). Inserting values of $b_{2}$, $b_{1}$, $b_{0}$ computed in Eqs. (\ref%
{13a})-(\ref{13c}) in Eq. (\ref{12d}), after simplification, we result with%
\begin{equation}
\left( e_{2}+e_{1}+e_{0}-1\right) y(t_{0})=0.  \tag{4.5}  \label{15a}
\end{equation}%
Since, $t_{0}$ is any initial state for non-zero initial conditions $%
y_{A}(t_{0})=y_{B}(t_{0})=y(t_{0})\neq 0,$ ($y_{B}^{\prime }(t_{0})$ may be
zero if $a_{0}=1$ due to Eq. (\ref{12b}), Eq. (\ref{15a}) implies that
\begin{equation}
e_{2}+e_{1}+e_{0}=1.  \tag{4.6}  \label{15b}
\end{equation}%
If commutativity with non-zero initial conditions is to be satisfied. Hence,
Eq. (\ref{12d}) can be relaced by%
\begin{equation}
y_{B}(t_{0})=y_{A}(t_{0})\neq 0  \tag{4.7}  \label{16a}
\end{equation}%
\begin{equation}
y_{B}^{\prime }(t_{0})=\frac{1-a_{0}}{a_{1}}y_{A}(t_{0}),  \tag{4.8}
\label{16b}
\end{equation}%
\begin{equation}
e_{2}+e_{1}+e_{0}=1,  \tag{4.9}  \label{16c}
\end{equation}%
\begin{equation}
y^{\prime \prime }=\frac{(1-a_{0})(1-a_{0}-a_{1}^{\prime })-a_{0}^{\prime
}a_{1}}{a_{1}^{2}}y_{A}(t_{0}).  \tag{4.10}  \label{16d}
\end{equation}

\section{Decomposition Formulas}

We now express the coefficients of the decompositions $A$ and $B$ in terms
of these of the decomposed system $C$. Comparing Eqs. (\ref{1}) and (\ref{5a}%
), equating the coefficients of third derivatives, and using Eq. (\ref{14}),
we have

\begin{equation}
a_{1}b_{2}=c_{3}=a_{1}e_{1}a_{1}^{2}\rightarrow a_{1}=\left( \frac{c_{3}}{%
e_{2}}\right) ^{1/3}.  \tag{5.1}  \label{17a}
\end{equation}

Comparing Eqs. (\ref{1}) and (\ref{5a}), equating the coefficients of second
derivatives, and using Eq. (\ref{14}), we obtain

\begin{equation*}
a_{1}b_{2}^{\prime }+a_{1}b_{1}+a_{0}b_{2}=c_{2}\rightarrow a_{0}=\frac{1}{%
b_{2}}\left( c_{2}-a_{1}b_{2}^{\prime }-a_{1}b_{1}\right)
\end{equation*}

\begin{equation*}
=\frac{1}{e_{2}a_{1}^{2}}\left\{ c_{2}-a_{1}\frac{d}{dt}\left(
e_{2}a_{1}^{2}\right) -a_{1}e_{2}[(a_{1}^{\prime
}+2a_{0})a_{1}+e_{1}a_{1}]\right\}
\end{equation*}

\begin{equation*}
=\frac{1}{e_{2}a_{1}^{2}}\left[ c_{2}-2e_{2}a_{1}^{2}a_{1}^{\prime
}-e_{2}a_{1}^{2}(a_{1}^{\prime }+2a_{0})-a_{1}^{2}\right]
\end{equation*}

\begin{equation*}
=\frac{1}{e_{2}}\left( \frac{c_{2}}{a_{1}^{2}}-2e_{2}a_{1}^{\prime
}-e_{2}a_{1}^{\prime }-2e_{2}a_{0}-e_{1}\right) =\frac{c_{2}}{e_{2}a_{1}^{2}}%
-3a_{1}^{\prime }-2a_{0}-\frac{e_{1}}{e_{2}}
\end{equation*}

\begin{equation*}
3a_{0}=\frac{c_{2}}{e_{2}a_{1}^{2}}-3a_{1}^{\prime }-\frac{e_{1}}{e_{2}}.
\end{equation*}

Dividing by $3$ and using Eq. (\ref{17a}), we proceed as

\begin{equation*}
a_{0}=\frac{c_{2}}{e_{2}a_{1}^{2}}-a_{1}^{\prime }-\frac{e_{1}}{3e_{2}}=%
\frac{(e_{2})^{1/3}c_{2}}{3e_{2}(c_{3})^{2/3}}-\frac{1}{3}\left( \frac{c_{3}%
}{e_{2}}\right) ^{-2/3}\frac{c_{3}^{\prime }}{e_{2}}-\frac{e_{1}}{3e_{2}}
\end{equation*}

\begin{equation}
=\frac{c_{2}}{3e_{2}^{1/3}c_{3}^{2/3}}-\frac{c_{3}^{\prime }}{%
3c_{3}^{2/3}e_{2}^{1/3}}-\frac{e_{1}}{3e_{2}}=\frac{c_{2}-c_{3}^{\prime }}{%
3e_{2}^{1/3}c_{3}^{2/3}}-\frac{e_{1}}{3e_{2}}.  \tag{5.2}  \label{17b}
\end{equation}%
Having computing $a_{1}$ and $a_{0}$ in Eqs. (\ref{17a}) and (\ref{17b}),
inserting those values in Eq. (\ref{14}), we compute $b_{2}$, $b_{1}$, $%
b_{0} $ and the results:

\begin{equation}
b_{2}=3e_{2}^{1/3}c_{3}^{2/3},  \tag{5.3}  \label{18a}
\end{equation}%
\begin{equation}
b_{1}=\frac{1}{3}\left[ \left( \frac{e_{2}}{c_{3}}\right) ^{1/3}\left(
2c_{2}-c_{3}^{\prime }\right) +e_{1}\left( \frac{c_{3}}{e_{2}}\right) ^{1/3}%
\right] ,  \tag{5.4}  \label{18b}
\end{equation}%
\begin{equation*}
b_{0}=\frac{1}{9}\left[ \left( \frac{e_{2}}{c_{3}}\right) ^{1/3}\left(
3c_{2}^{\prime }-3c_{3}^{\prime \prime }\right) +\frac{c_{2}^{2}+\left(
c_{3}^{\prime }\right) ^{2}-4c_{2}c_{3}^{\prime }}{c_{3}}\right]
\end{equation*}%
\begin{equation}
+\frac{1}{9}\left[ \frac{e_{1}(c_{2}-c_{3}^{\prime })}{e_{2}^{1/3}c_{3}^{2/3}%
}-\frac{2e_{1}^{2}}{e_{2}}\right] +e_{0}.  \tag{5.5}  \label{18c}
\end{equation}

Comparing Eq. (\ref{1}) with Eq. (\ref{5a}), two additional equations should
be satisfied for the equivalence of $C$ and $AB$ (or $BA$, since $AB$ is a
commutative pair). These are%
\begin{equation}
c_{1}=a_{1}b_{1}^{\prime }+a_{1}b_{0}+a_{0}b_{1},  \tag{5.6}  \label{19a}
\end{equation}%
\begin{equation}
c_{0}=a_{1}b_{0}^{\prime }+a_{0}b_{0}.  \tag{5.7}  \label{19b}
\end{equation}

Inserting the valves of $a_{1}$, $a_{0}$ in\ Eqs. (\ref{17a}), (\ref{17b})
and $b_{1}$, $b_{0}$ as computed in Eqs. (\ref{18b}), (\ref{18c}) into Eqs. (%
\ref{19a}) and (\ref{19b}) and making a grate deal of computations, we
obtain the additional conditions to be satisfied;
\begin{equation}
c_{1}=\left( c_{2}^{\prime }-\frac{2}{3}c_{3}^{\prime \prime }\right) +\frac{%
1}{c_{3}}\left[ \frac{5}{9}\left( c_{3}^{\prime }\right)
^{2}-c_{2}c_{3}^{\prime }+\frac{c_{2}^{2}}{3}\right] +c_{3}^{1/3}\frac{1}{%
e_{2}^{1/3}}\left( e_{0}-\frac{e_{1}^{2}}{3e_{2}}\right) ,  \tag{5.8}
\label{20a}
\end{equation}%
\begin{equation*}
c_{0}=\frac{1}{3}\left( c_{2}^{\prime \prime }-c_{3}^{\prime \prime \prime
}\right) +\frac{1}{3c_{3}}(c_{2}-2c_{3}^{\prime })(c_{2}^{\prime
}-c_{3}^{\prime \prime })
\end{equation*}%
\begin{equation*}
+\frac{1}{27c_{3}^{2}}\left[ 15\left( c_{3}^{\prime }\right)
^{2}-8c_{3}^{\prime }c_{2}-6c_{3}^{\prime \prime }c_{2}+c_{2}^{2}\right]
(c_{2}-c_{2}^{\prime })
\end{equation*}%
\begin{equation}
+\frac{c_{2}-c_{3}^{\prime }}{3c_{3}^{2/3}b_{2}e_{2}^{1/3}}\left( e_{0}-%
\frac{e_{1}^{2}}{3e_{2}}\right) +\frac{e_{1}}{3}\left( \frac{2e_{1}^{2}}{9}-%
\frac{e_{0}}{e_{2}}\right) .  \tag{5.9}  \label{20b}
\end{equation}

In the light of the result obtained so far, we now express the main theorem
about the decomposition of a third-order linear time-varying system into its
commutative first and second-order linear time-varying components.

\begin{theorem}
The necessary and sufficient conditions that a third-order linear
time-varying system described by Eq. (\ref{1}) into its cascade connected
linear time-varying commutative pairs of first-order and second-order are
that

\begin{description}
\item[i)] The coefficient $c_{1}$ and $c_{0}$ be expressible in terms of $%
c_{3}$ and $c_{2}$ through formulas Eqs. (\ref{20a}) and (\ref{20b}) where $%
e_{2}$, $e_{1}$, $e_{0}$ are some constants.

\item[ii)] If the condition $y(t_{0})$ of $C$ is different from zero,
additional necessary and sufficient condition are expressed as%
\begin{equation}
e_{2}+e_{1}+e_{0}=1,  \tag{5.10}  \label{21a}
\end{equation}%
\begin{equation}
y^{\prime }(t_{0})=\left[ \left( \frac{e_{2}}{c_{3}}\right) ^{1/3}\left( 1+%
\frac{e_{1}}{3e_{2}}\right) -\frac{c_{2}-c_{3}^{\prime }}{3c_{3}}\right]
y(t_{0})\text{ at }t=t_{0},  \tag{5.11}  \label{21b}
\end{equation}%
\begin{equation*}
y^{\prime }(t_{0})=\left[ \left( \frac{e_{2}}{c_{3}}\right) ^{1/3}\left( 1+%
\frac{e_{1}}{3e_{2}}\right) -\frac{c_{2}-c_{3}^{\prime }}{3c_{3}}\right]
^{2}y_{A}(t_{0})
\end{equation*}%
\begin{equation}
+\frac{d}{dt}\left[ \left( \frac{e_{2}}{c_{3}}\right) ^{1/3}\left( 1+\frac{%
e_{1}}{3e_{2}}\right) -\frac{c_{2}-c_{3}^{\prime }}{3c_{3}}\right]
y_{A}(t_{0}).  \tag{5.12}  \label{21c}
\end{equation}
\end{description}
\end{theorem}

\begin{proof}
Part i) is simply re-expressing of Eqs. (\ref{20a}) and (\ref{20b}). Eq. (%
\ref{21a}) is the repetition of Eq. (\ref{16c}). Eq. (\ref{21b}) is obtained
from Eq. (\ref{16b}) by inserting in the valves of $a_{1}$ and $a_{0}$ in
Eqs. (\ref{17a}) and (\ref{17b}), respectively. So,
\end{proof}

\begin{equation*}
y^{\prime }(t_{0})=\frac{1-a_{0}}{a_{1}}y(t_{0})=\left( \frac{e_{2}}{c_{3}}%
\right) ^{1/3}\left( 1+\frac{e_{1}}{3e_{2}}-\frac{c_{2}-c{^{\prime }}_{3}}{%
3e_{2}^{1/3}c_{3}^{2/3}}\right) y(t_{0})
\end{equation*}

\begin{equation*}
=\left[ \left( \frac{e_{2}}{c_{3}}\right) ^{1/3}\left( 1+\frac{e_{1}}{3e_{2}}%
\right) -\frac{c_{2}-c_{3}^{\prime }}{3c_{3}}\right] y(t_{0}).
\end{equation*}

Finally, Eq. (\ref{21c}) is obtained from From Eq. (\ref{16c}) by inserting
in valves of $a_{1}$ and $a_{0}$ in Eqs. (\ref{17a}) and (\ref{17b}),
respectively as follows:%
\begin{equation*}
y^{\prime \prime }(t_{0})=\frac{\left( 1-a_{0}\right) \left(
1-a_{0}-a_{1}^{\prime }\right) -a_{0}^{\prime }a_{1}}{a_{1}^{2}}y_{A}(t_{0})
\end{equation*}%
\begin{equation*}
=\frac{(1-a_{0})^{2}-a_{1}^{\prime }+a_{0}a_{1}^{\prime }-a_{0}^{\prime
}a_{1}}{a_{1}^{2}}y_{A}(t_{0})
\end{equation*}

\begin{equation*}
=\left[ \frac{(1-a_{0})^{2}}{a_{1}^{2}}+\frac{d}{dt}\frac{1}{a_{1}}-\frac{d}{%
dt}\frac{a_{0}}{a_{1}}\right] y_{A}(t_{0})=\left[ \frac{(1-a_{0})^{2}}{%
a_{1}^{2}}+\frac{d}{dt}\frac{1-a_{0}}{a_{1}}\right] y_{A}(t_{0})
\end{equation*}%
\begin{equation*}
=\left[ \left( \frac{e_{2}}{c_{3}}\right) ^{1/3}\left( 1+\frac{e_{1}}{3e_{2}}%
\right) -\frac{c_{2}-c_{3}^{\prime }}{3c_{3}}\right] ^{2}y_{A}(t_{0})
\end{equation*}%
\begin{equation*}
+\frac{d}{dt}\left[ \left( \frac{e_{2}}{c_{3}}\right) ^{1/3}\left( 1+\frac{%
e_{1}}{3e_{2}}\right) -\frac{c_{2}-c_{3}^{\prime }}{3c_{3}}\right]
y_{A}(t_{0}).
\end{equation*}

The second theorem expresses how to obtain the commutative pairs of
decomposition $A$ and $B$.

\begin{theorem}
For a third-order linear time-varying system $C$ described by Eq.(\ref{1})
with the conditions of Theorem I satisfied, the decomposed commutative pairs
$A$ and $B$ are found by Eqs. (\ref{17a}) and (\ref{17b}) for the
coefficients $a_{1}$ and $a_{0}$ of $A$, and by Eqs. (\ref{18a}), (\ref{18b}%
), (\ref{18c}) for the coefficients of $b_{2}$ , $b_{1}$, $b_{0}$ of $B$,
all respectively. Further, for the commutative decompositions with non-zero
initial condition $y(t_{0})\neq 0$, Eq. (\ref{21a}) relating the constants $%
e_{2}$, $e_{1}$, $e_{0}$ must be satisfied; and $y_{B}^{\prime
}(t_{0})=y^{\prime }(t_{0})$ and $y^{\prime \prime }(t_{0})$ must be
expressible in terms of $y(t_{0})=y_{B}(t_{0})=y_{A}(t_{0})$ as in Eqs. (\ref%
{21b}) and (\ref{21c}), respectively.
\end{theorem}

\begin{proof}
The proof follows from the development of the mentioned equation. Equality
of $y(t_{0})=y_{B}(t_{0})=y_{A}(t_{0})$ is a result of Eqs. (\ref{5b}) and (%
\ref{6a}); equality of $y^{\prime }(t_{0})=y_{A}^{\prime }(t_{0})$ is
already expressed in Eq. (\ref{5b}).
\end{proof}

\section{\protect\bigskip Examples}

In this section, four examples are considered to illustrate the results of
the paper. The simulations are conducted by MATLAB R2012a and obtained by a
PC Intel\registered\ Core\texttrademark\ i3 CPV, 2.13 GHz, 3.86 GB of RAM
well verify the results.

\subsection{\protect\bigskip Example 1}

Let $C$ be the third-order linear time-varying system defined by%
\begin{equation}
y^{\prime \prime \prime }(t)+(t+1)y^{\prime \prime }(t)+\frac{1}{3}%
(t^{2}+2t)y^{\prime }(t)+\frac{1}{27}\left( t^{3}+3t^{2}+9\right) y(t)=x(t),
\tag{6.1}  \label{22a}
\end{equation}%
which the coefficients are%
\begin{equation}
c_{3}=1,\text{ }c_{2}=(t+1),\text{ }c_{1}=\frac{1}{3}(t^{2}+2t),\text{ }%
c_{0}=\frac{1}{27}(t^{3}+3t^{2}+9),  \tag{6.2}  \label{22b}
\end{equation}%
with the constants%
\begin{equation}
e_{2}=e_{1}=1,\text{ }e_{0}=-1,  \tag{6.3}  \label{22c}
\end{equation}%
which satisfies Eq. (\ref{21a}), it is true the conditions $i)$ of Theorem I
are satisfied; that is $c_{1}$ and $c_{0}$ satisfy Eqs. (\ref{20a}) and (\ref%
{20b}), respectively. For the validity of the decomposition with non-zero
initial condition $y(t_{0})\neq 0,$ condition (\ref{21a}) of $ii)$ is
satisfied by the constant chosen in Eq. (\ref{22c}). Further, Eqs. (\ref{21b}%
) and (\ref{21c}) of $ii)$ together with Theorem 2 yield

\begin{equation}
y_{A}(t_{0})=y_{B}(t_{0})=y(t_{0}),  \tag{6.4}  \label{23a}
\end{equation}%
\begin{equation*}
y_{B}^{\prime }(t_{0})=y^{\prime }(t_{0})=\left[ \left( \frac{1}{1}\right)
^{1/3}\left( 1+\frac{1}{3}\right) -\frac{t_{0}+1}{3}\right] y(t_{0})
\end{equation*}%
\begin{equation}
=\left( 1-\frac{t_{0}}{3}\right) y(t_{0})=y(t_{0})\text{ for }t_{0}=0,
\tag{6.5}  \label{23b}
\end{equation}%
\begin{equation*}
y^{\prime \prime }(t_{0})=\left[ \left( 1+\frac{1}{3}-\frac{t_{0}+1}{3}%
\right) ^{2}+\frac{d}{dt}\left. \left( 1+\frac{1}{3}-\frac{t+1}{3}\right)
\right\vert _{t=t_{0}}\right] y(t_{0})
\end{equation*}%
\begin{equation*}
=\left[ \left( 1-\frac{t_{0}}{3}\right) ^{2}+\frac{d}{dt}\left. \left( 1-%
\frac{t}{3}\right) \right\vert _{t=t_{0}}\right] y(t_{0})
\end{equation*}%
\begin{equation}
=\left[ \left( 1-\frac{t_{0}}{3}\right) ^{2}-\frac{1}{3}\right] y(t_{0})=%
\frac{2}{3}y(t_{0})\text{ for }t_{0}=0.  \tag{6.6}  \label{23c}
\end{equation}

Theorem 2 yields the following coefficients for decomposed subsystem $A$ and
$B:$

\begin{equation}
A:y_{A}^{\prime }(t)+\frac{t}{3}y_{A}(t)=x_{A}(t),  \tag{6.7}  \label{24a}
\end{equation}%
\begin{equation}
B:y_{B}^{\prime \prime }(t)+\frac{2t+3}{3}y_{B}^{\prime }(t)+\frac{t^{2}+3t-6%
}{9}y_{B}(t)=x_{B}(t).  \tag{6.8}  \label{24b}
\end{equation}%
Simulations are carried out with a sinusoidal input of amplitude $10$, bias $%
-5$ and frequency $3$. Fixed step length of $0.01$ is used by
ode(Bogacki-Shampine). Simulink results of MATLAB R2012 are shown in Fig. 2.
The initial time $t_{0}$ is assumed $1$ and the initial states are taken as $%
y(1)=y_{A}(1)=y_{B}(1)=1$. When $y_{B}^{\prime }(1)=y^{\prime }(1)=1$ and $%
y^{\prime \prime }(1)=2/3$ as implied by (\ref{23a}), (\ref{23b}), (\ref{23c}%
) all the decomposition conditions are satisfied and $AB,$ $BA,$ and $C$
give the same responses as indicated by the figure legend. But when $%
y^{\prime \prime }(1)$ is changed to $2$ which does not satisfy (\ref{23c}),
the response $C1$ becomes different from those of $AB$ and $BC$, that is the
decomposition get spoiled; although $A$ and $B$ are commutative, they are
not the correct decomposition of $C$. On the other hand, when $y_{B}^{\prime
}(1)$ is made $-1$, that is (\ref{23b}) is not satisfied, the response of $%
AB $ (indicated by $AB3$) gets different from those of $BA$ and $C$, so
commutative decomposition of $C$ into $A$ and $B$ is not valid again.

\begin{figure}[tbp]
\includegraphics[width=12cm, height=12cm]{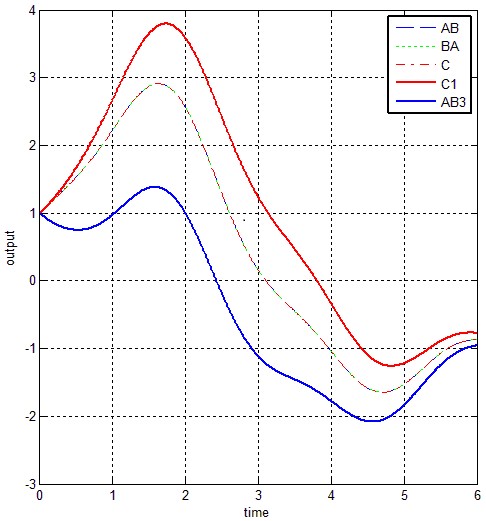}
\caption{Decomposition of $C$ into its commutative pairs $A$ and $B $ ($%
AB,BA,C$); some of the conditions of decomposition are not satisfied ($C1,$ $%
AB3$)}
\label{.}
\end{figure}

\subsection{Example 2}

Consider $C$ defined by%
\begin{equation}
t^{3}y^{\prime \prime \prime }(t)+7t^{2}y^{\prime \prime }(t)+9ty^{\prime
}(t)+y(t)=x(t),  \tag{6.9}  \label{25.1}
\end{equation}%
which satisfies the condition of Theorem I with $e_{2}=e_{1}=1,$ $e_{0}=-1.$
Hence, with $c_{3}=t^{3},$ $c_{2}=7t^{2},$ the initial conditions should
satisfy Eqs. (\ref{21b}) and (\ref{21c}):%
\begin{equation}
y^{\prime }(t_{0})=\left[ \frac{1}{t_{0}}\left( 1+\frac{1}{3}\right) -\frac{%
7t_{0}^{2}-3t_{0}^{2}}{3t_{0}^{3}}\right] y(t_{0})=0,  \tag{6.10}
\label{25.2}
\end{equation}%
\begin{equation*}
y^{\prime \prime }(t_{0})=\left\{ \left[ \frac{1}{t_{0}}\left( 1+\frac{1}{3}%
\right) -\frac{7t_{0}^{2}-3t_{0}^{2}}{3t_{0}^{3}}\right] ^{2}\right.
\end{equation*}%
\begin{equation}
\left. +\frac{d}{dt}\left. \left[ \frac{1}{t}\left( 1+\frac{1}{3}\right) -%
\frac{7t^{2}-3t^{2}}{3t^{3}}\right] \right\vert _{t=t_{0}}\right\}
y(t_{0})=0.  \tag{6.11}  \label{25.3}
\end{equation}

The decompositions $A$ and $B$ are found by using Eqs. (\ref{17a}), (\ref%
{17b}) and (\ref{18a})-(\ref{18c}) as%
\begin{equation}
A:ty_{A}^{\prime }(t)+y_{A}(t)=x_{A}(t),\text{ }y_{A}(t_{0})=y\left(
t_{0}\right) ,  \tag{6.12}  \label{25.4}
\end{equation}%
\begin{equation}
B:t^{2}y_{B}^{\prime \prime }(t)+4ty_{B}^{\prime
}(t)+y_{B}(t)=x_{B}(t),y_{B}(t_{0})=y\left( t_{0}\right) ,\text{ }%
y_{B}^{\prime }(t_{0})=y^{\prime }\left( t_{0}\right) =0.  \tag{6.13}
\label{25.5}
\end{equation}

Note that $y^{\prime }(t_{0})=y_{B}^{\prime }\left( t_{0}\right) ,$ $%
y^{\prime \prime }(t_{0})$ are zero for all initial times $t_{0}.$ The
simulations are carried out with a sinusoidal input of amplitude $100$,
frequency $100$ Hz and phase $\pi /3$ rad; the initial time $t_{0}$ is taken
as $0.01$ and stop time is $0.15$; ode(Bogacki-Shampine) solver is used with
step-length of $0.001.$ The initial values are assumed as $%
y_{A}(t_{0})=y_{B}(t_{0})=y(t_{0})=-4$. As it is seen in Fig. 3, $C$ and its
commutative decompositions $AB$ and $BA$ yield the same responses (see $%
AB=BA=C$). When the decomposition requirement on initial condition get
spoiled, that is $y^{\prime }(0.01)=y_{B}^{\prime }(0.01)\neq 0$ and taken
as $-100$, the decomposition is not valid at all as seen from plots $AB1,$ $%
BA1,$ $C1$ in the figure. It is important to note that the cascade
connection $BA$ is least affected from this change. Hence, it is preferable
decomposition or synthesis of $C$ when compared with $AB$ as far as
sensitivity to initial conditions is concerned.

\begin{figure}[tbp]
\includegraphics[width=12cm, height=12cm]{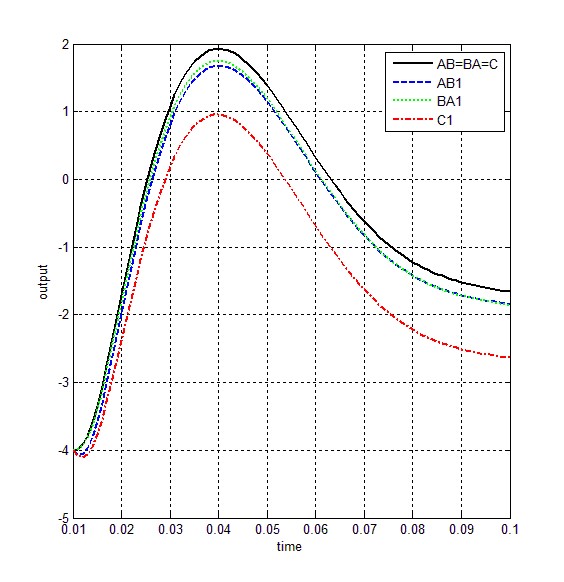}
\caption{Simulation result for Example 2 when decomposition requirements are
satisfied ($AB=BA=C$) and they are not satisfied ($AB1,BA1,C1$)}
\label{.}
\end{figure}

\subsection{Example 3}

Let $C$ be the third-order Euler system defined by%
\begin{equation}
t^{3}y^{\prime \prime \prime }+9t^{2}y^{\prime \prime }+\frac{53}{3}%
ty^{\prime }+\frac{155}{27}y=x.  \tag{6.14}  \label{25a}
\end{equation}%
Comparing it with (\ref{1}), its coefficients are%
\begin{equation}
c_{3}=t^{3},\text{ }c_{2}=9t^{2},\text{ }c_{1}=\frac{53}{3}t,\text{ }c_{0}=%
\frac{155}{27}.  \tag{6.15}  \label{25b}
\end{equation}

It is true that the choice $e_{2}=e_{1}=1,$ $e_{0}=-1$ satisfy the
conditions of Theorem I; that is Eqs. (\ref{20a}), (\ref{20b}) and (\ref{21a}%
) are satisfied. Hence, the decomposition into first and second-order
commutative pairs with non-zero initial condition $y(t_{0})\neq 0$ is
possible. The initial condition of $A$ and $B$ as well as those of $C$ are
found by using Theorem I and II. In fact,%
\begin{equation}
y_{A}(t_{0})=y_{B}(t_{0})=y(t_{0})\neq 0,  \tag{6.16}  \label{26a}
\end{equation}%
\begin{equation}
y_{B}^{\prime }(t_{0})=y^{\prime }(t_{0})=\left[ \left( \frac{1}{t_{0}^{3}}%
\right) ^{1/3}\left( 1+\frac{1}{3}\right) -\frac{9t_{0}^{2}-3t_{0}^{2}}{%
3t_{0}^{2}}\right] y(t_{0})=-\frac{2}{3t_{0}}y(t_{0}),  \tag{6.17}
\label{26b}
\end{equation}%
\begin{equation*}
y^{\prime \prime }(t_{0})=\left\{ \left[ \frac{1}{t_{0}}\left( 1+\frac{1}{3}%
\right) -\frac{3t_{0}^{2}-3t_{0}^{2}-1}{3t_{0}^{2}}\right] ^{2}\right.
\end{equation*}%
\begin{equation*}
\left. +\frac{d}{dt}\left. \left( \frac{1}{t}\frac{4}{3}-\frac{3t^{2}-3t^{2}%
}{3t^{2}}\right) \right\vert _{t=t_{0}}\right\} y(t_{0})
\end{equation*}%
\begin{equation}
=\left( \frac{4}{9t_{0}^{2}}+\frac{2}{3t_{0}^{2}}\right) y(t_{0})=\frac{10}{%
9t_{0}^{2}}y(t_{0}).  \tag{6.18}  \label{26c}
\end{equation}

The decompositions $A$ and $B$ are found by using the coefficients given in
Eqs. (\ref{17a}), (\ref{17b}) and (\ref{18a})-(\ref{18c}). With the above
initial conditions $A$ and $B$ are defined by%
\begin{equation}
A:ty_{A}^{\prime }(t)+\frac{5}{3}y_{A}(t)=x_{A}(t);\text{ }%
y_{A}(t_{0})=y(t_{0}),  \tag{6.19}  \label{27a}
\end{equation}%
\begin{equation*}
B:t^{2}y_{B}^{\prime \prime }(t)+\frac{16}{3}ty_{B}^{\prime }(t)+\frac{31}{9}%
y_{B}(t)=x_{B}(t);
\end{equation*}%
\begin{equation}
\text{ }y_{B}(t_{0})=y(t_{0});\text{ }y_{B}^{\prime }(t_{0})=y^{\prime
}(t_{0})=-\frac{2}{3t_{0}^{2}}y(t_{0}).  \tag{6.20}  \label{27b}
\end{equation}%
\qquad

The simulations are done for sinusoidal input of amplitude $10$ and
frequency $1$. The initial time $t_{0}=1;$ ode3(Bogacki-Shampine) solver is
used with a fixed step-length of $0.01$; simulations are stopped at $t=10$.
When the initial conditions $y^{\prime \prime }(1)=10/9,$ $y^{\prime
}(1)=y_{B}^{\prime }(1)=-2/3$ are chosen in accordance with $%
y_{A}(1)=y_{B}(1)=y(1)=1$ as to satisfy the decomposition above mentioned
conditions, $AB,$ $BA,$ $C$ give the same response as shown in Fig. 4 (see $%
AB=BA=C$). In the same figure, zero input responses ($AB1=BA1=C1$) and zero
state responses ($AB2=BA2=C2$) are also potted. Obviously, decomposition is
valid for unexited-unrelaxed and exited-relaxed cases as well.

\begin{figure}[tbp]
\includegraphics[width=12cm, height=12cm]{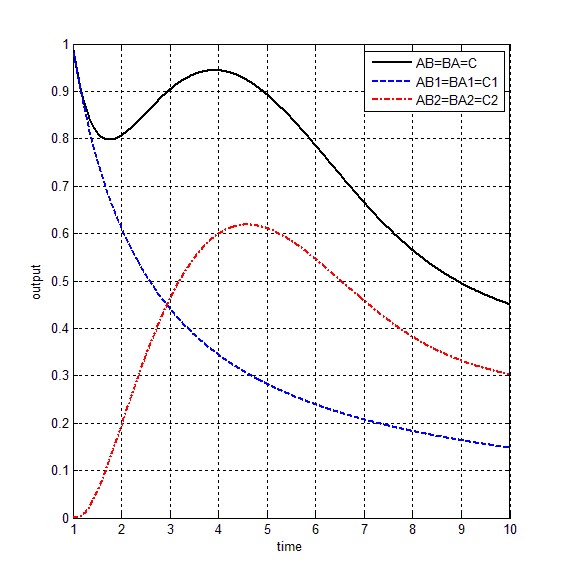}
\caption{Complete zero-input ($C1$), and zero-state ($C2$) responses of
Example 3}
\label{.}
\end{figure}

\subsection{Example 4}

This example is the same as the first one except all the initial conditions
are taken as zero and a noise signal is added between the junction of
subsystems $A$ and $B$. The noise is a pulse sequence with amplitude $4$, $%
\% $ $50$ pulse with, and a bias of $-2.3$. The simulation results are shown
in Fig. 5. Obviously, the interconnection $AB$ is less effected by this
noise than $BA$ connection when compared with the output of the original
system $C$. Hence, the cascade synthesis $AB$ should be preferred rather
than $BA$.

\begin{figure}[tbp]
\includegraphics[width=12cm, height=12cm]{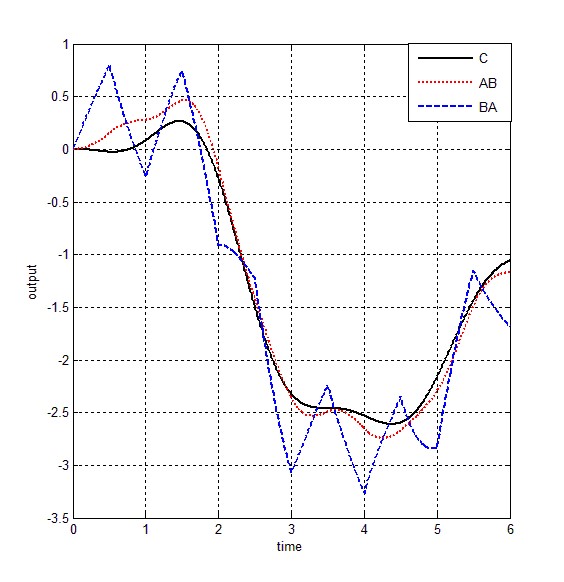}
\caption{Outputs of the original system $C$ and its cascade decompositions $%
AB$ and $BA$ when disturbance exists at the interconnection}
\label{.}
\end{figure}

\section{\protect\bigskip Conclusions}

In this paper, the decomposition of any third-order linear time-varying
system into its first and second-order commutative pairs is investigated.
Explicit decomposition formulas are derived for the case of zero and
non-zero initial conditions. The results are validated by computer
simulations. The work is original and appears for the first time in the
literature. It is important from the synthesis and/or design point of views
of engineering systems. Many design methods are based on tearing and
reconstruction, which is combining simple components to obtain an assembly.
Further, it is shown that some combinations may be better than the others
when sensitivity to initial conditions and noise disturbance at the
interconnection is taken into account. On the other hand, commutativity of
cascade connected systems have gained a grade deal of interest and its
possible benefits have been pointed out on the literature. Hence, the
results of this paper can be used readily for beneficial synthesis of
third-order linear time-varying systems.

\end{document}